\documentclass[12pt,a4paper]{article}
\usepackage{ae,lmodern} 
\usepackage[english]{babel}
\usepackage[utf8]{inputenc}
\usepackage[T1]{fontenc}
\usepackage{amsmath}
\usepackage{amsfonts}
\usepackage{amssymb}
\usepackage{makeidx}
\usepackage{amsthm}
\usepackage{algorithm}
\usepackage{algorithmicx}
\usepackage{algpseudocode}
\usepackage{graphicx}
\usepackage[dvipsnames]{xcolor}
\usepackage{tikz}
\usepackage{eso-pic}

\newtheorem{theorem}{Theorem}[section]
\newtheorem{lemma}[theorem]{Lemma}

\newcommand{\sm}{\setminus}

\theoremstyle{definition}

  \author{Myriam Preissmann\thanks{Univ. Grenoble Alpes, CNRS, Grenoble INP, G-SCOP, 38000 Grenoble, France}~, Cléophée  Robin\footnotemark[1]~\thanks{Univ Lyon, EnsL, UCBL, CNRS, LIP,
    F-69342, LYON Cedex 07, France. Partially
    supported by the LABEX MILYON (ANR-10-LABX-0070) of Universit\'e
    de Lyon, within the program ‘‘Investissements d'Avenir’’
    (ANR-11-IDEX-0007) operated by the French National Research Agency     (ANR)}~~and Nicolas Trotignon\footnotemark[2]}
\title{On the complexity of colouring antiprismatic graphs}

\begin{document}

\maketitle


\begin{abstract}
A graph $G$ is \emph{prismatic} if for every triangle $T$ of $G$,
  every vertex of $G$ not in $T$ has a unique neighbour in $T$. The
  complement of a prismatic graph is called \emph{antiprismatic}. The
  complexity of colouring antiprismatic graphs is still
  unknown. Equivalently, the complexity of the clique cover problem in
  prismatic graphs is not known.
  
  Chudnovsky and Seymour gave a full structural description of
  prismatic graphs.  They showed that the class can be divided into
  two subclasses: the orientable prismatic graphs, and the
  non-orientable prismatic graphs.  We give a polynomial time
  algorithm that solves the clique cover problem in every
  non-orientable prismatic graph.  It relies on the the structural description and on later work of Javadi and Hajebi.  We give a polynomial time algorithm which 
  solves the vertex-disjoint triangles problem for every prismatic
  graph. It does not rely on the structural description.
\end{abstract}

\section{Introduction}

We consider undirected simple graphs (no multiple edge, no
loop). A \emph{$k$-colouring of a graph $G$} is a function
$f$ from $V(G)$ to $\{1, \dots, k\}$ such that for all $uv\in E(G)$,
$f(u)\neq f(v)$.  The colouring problem is the problem whose input is
a graph $G$ and whose output is a $k$-colouring of $G$ such that $k$ is
minimum.  The colouring problem is NP-hard and is the object of much
attention. In particular, its complexity is known for several classes
of graphs.  

When $H$ is a graph, a graph $G$ is \emph{$H$-free} if it does not
contain any induced subgraph isomorphic to $H$.  When $\cal H$ is a set
of graphs, $G$ is \emph{$\cal H$-free} if $G$ is $H$-free for all
$H\in \cal H$.  We denote by $P_k$ the path on $k$ vertices and by
$C_k$ the cycle on $k$ vertices. When $G$ and $H$ are graphs, we
denote by $G+H$ their disjoint union. When $G$ is a graph and $k$ an
integer, we denote by $kG$ the disjoint union of $k$ copies of $G$. We
denote by $K_k$ the complete graph on $k$ vertices.  We denote by
$K_{k, l}$ the complete bipartite graph with one side of the
bipartition of size $k$ and the other side of size $l$. A \textit{stable set} is a set of vertices that are pairwise nonadjacent. A \textit{clique} is a set of vertices that are pairwise adjacent. 

Král', Kratochvíl, Tuza and
Woeginger~\cite{daniel_kral_complexity_2001} proved the following
dichotomy: the coloring problem for $H$-free graphs is polynomial time
solvable if $H$ is an induced subgraph of $P_4$ or an induced subgraph
of $K_1+P_3$, and NP-hard otherwise.  This motivated the systematic
study of the colouring problem restricted to $\{H_1, H_2\}$-free
graphs for all possible pairs of graphs $H_1$, $H_2$, or even to
$\cal H$-free graph in general.  The complexity has been shown to be
polynomial or NP-hard for all sets of graphs on at most four vertices
(see Figure~\ref{f:g4}), apart for the following four cases that are
still open (see~\cite{v.v._lozin_vertex_2015} and~\cite{golovach_coloring_2012}).

\begin{itemize}
\item ${\cal H = }\{C_4, 4K_1\}$
\item ${\cal H = } \{K_{1, 3}, 4K_1\}$
\item ${\cal H = }\{K_{1, 3}, 2P_1+P_2\}$
\item ${\cal H = }\{K_{1, 3}, 2P_1+P_2, 4K_1\}$
\end{itemize}

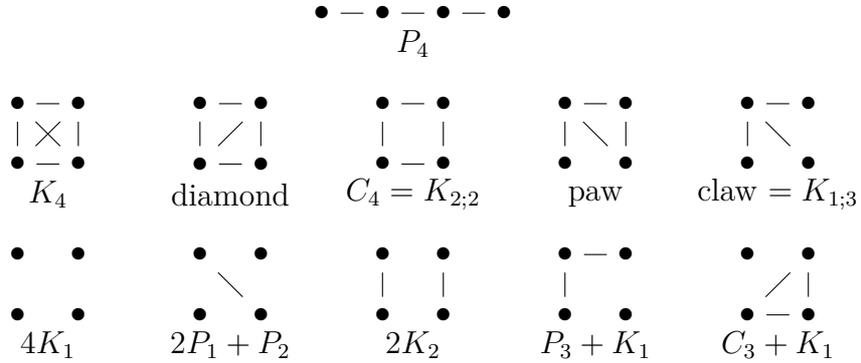
\begin{figure}
\centering
	\begin{tikzpicture}[scale=0.8]
	\node[-] (a) at (-5,1.5) {$\bullet$};
\node[-](b) at (-6,1.5) {$\bullet$};
\node[-] (c) at (-7,1.5) {$\bullet$};
\node[-] (d) at (-8,1.5) {$\bullet$};
\draw[-] (a) -- (b); 
\draw[-] (b) -- (c);
\draw[-] (c) -- (d); 
\node[-] (name) at (-6.5,1) {$P_4$};

\node[-] (a) at (0,0) {$\bullet$};
\node[-](b) at (0,-1) {$\bullet$};
\node[-] (c) at (-1,-1) {$\bullet$};
\node[-] (d) at (-1,0) {$\bullet$};
\draw[-] (a) -- (d); 
\draw[-] (b) -- (d);
\draw[-] (c) -- (d); 
\node[-] (name) at (-0.5,-1.5) {claw  $ = K_{1;3}$}; 

\node[-] (a) at (-3,0) {$\bullet$};
\node[-](b) at (-4,0) {$\bullet$};
\node[-] (c) at (-3,-1) {$\bullet$};
\node[-] (d) at (-4,-1) {$\bullet$};
\draw[-] (b) -- (c); 
\draw[-] (a) -- (b);  
\draw[-] (b) -- (d);
\draw[-] (c) -- (a);   
\node[-] (name) at (-3.5,-1.54) {paw};

\node[-] (a) at (-6,0) {$\bullet$};
\node[-](b) at (-7,-1) {$\bullet$};
\node[-] (c) at (-6,-1) {$\bullet$};
\node[-] (d) at (-7,0) {$\bullet$};
\draw[-] (a) -- (d);
\draw[-] (a) -- (c); 
\draw[-] (d) -- (b); 
\draw[-] (c) -- (b); 
\node[-] (name) at (-6.5,-1.5) {$C_4=K_{2;2}$};

\node[-] (a) at (-9,0) {$\bullet$};
\node[-](b) at (-10,-1) {$\bullet$};
\node[-] (c) at (-9,-1) {$\bullet$};
\node[-] (d) at (-10,0) {$\bullet$};
\draw[-] (a) -- (d);
\draw[-] (a) -- (c); 
\draw[-] (d) -- (b); 
\draw[-] (c) -- (b); 
\draw[-] (a) -- (b);
\node[-] (name) at (-9.5,-1.5) {diamond}; 

\node[-] (a) at (-12,0) {$\bullet$};
\node[-](b) at (-13,-1) {$\bullet$};
\node[-] (c) at (-12,-1) {$\bullet$};
\node[-] (d) at (-13,0) {$\bullet$};
\draw[-] (d) -- (b); 
\draw[-] (d) -- (c); 
\draw[-] (a) -- (d); 
\draw[-] (a) -- (b);
\draw[-] (a) -- (c);
\draw[-] (b) -- (c);   
\node[-] (name) at (-12.5,-1.5) {$K_4$};
\node[-] (a) at (0,-2.5) {$\bullet$};
\node[-](b) at (0,-3.5) {$\bullet$};
\node[-] (c) at (-1,-2.5) {$\bullet$};
\node[-] (d) at (-1,-3.5) {$\bullet$}; 
\draw[-] (b) -- (a);
\draw[-] (d) -- (a); 
\draw[-] (d) -- (b); 
\node[-] (name) at (-0.5,-4) {$C_3+K_1$}; 

\node[-] (a) at (-3,-2.5) {$\bullet$};
\node[-](b) at (-4,-2.5) {$\bullet$};
\node[-] (c) at (-3,-3.5) {$\bullet$};
\node[-] (d) at (-4,-3.5) {$\bullet$};
\draw[-] (b) -- (d); 
\draw[-] (a) -- (b);    
\node[-] (name) at (-3.5,-4) {$P_3+K_1$};

\node[-] (a) at (-6,-2.5) {$\bullet$};
\node[-](b) at (-7,-3.5) {$\bullet$};
\node[-] (c) at (-6,-3.5) {$\bullet$};
\node[-] (d) at (-7,-2.5) {$\bullet$};
\draw[-] (a) -- (c); 
\draw[-] (b) -- (d); 
\node[-] (name) at (-6.5,-4) {$2K_2$};

\node[-] (a) at (-9,-2.5) {$\bullet$};
\node[-](b) at (-10,-3.5) {$\bullet$};
\node[-] (c) at (-9,-3.5) {$\bullet$};
\node[-] (d) at (-10,-2.5) {$\bullet$};
\draw[-] (d) -- (c);
\node[-] (name) at (-9.5,-4) {$2P_1+P_2$}; 

\node[-] (a) at (-12,-2.5) {$\bullet$};
\node[-](b) at (-13,-3.5) {$\bullet$};
\node[-] (c) at (-12,-3.5) {$\bullet$};
\node[-] (d) at (-13,-2.5) {$\bullet$};
\node[-] (name) at (-12.5,-4) {$4K_1$};
\end{tikzpicture}	
\caption{All graphs on 4 vertices\label{f:g4}}
\end{figure}

Lozin and Malyshev~\cite{v.v._lozin_vertex_2015} noted that a
$\{K_{1, 3}, 2P_1+P_2\}$-free graph is either $4K_1$-free or has no
edges.  Therefore, $\{K_{1, 3}, 2P_1+P_2, 4K_1\}$-free graphs are
essentially equivalent to $\{K_{1, 3}, 2P_1+P_2\}$-free graphs, in the
sense that the complexity of the coloring problem is the same for both
classes.

In this paper, we study $\{K_{1, 3}, 2P_1+P_2, 4K_1\}$-free graphs.
They have been introduced in a different context.  Chudnovsky and
Seymour~(see~\cite{chudnovsky_claw-free_2007}) described the structure
of claw-free graphs in a series of articles (the \emph{claw} is another name
for $K_{1, 3}$).  The two first articles of the series are about the
so-called \emph{prismatic graphs}. Let us define them.  A
\emph{triangle} in a graph is a set of three pairwise adjacent
vertices.  A graph $G$ is \textit{prismatic} if for every triangle $T$
of $G$, every vertex of $G$ not in $T$ has a unique neighbour in
$T$. A graph is \emph{antiprismatic} if its complement is prismatic.
It is straightforward to check that antiprismatic graphs are precisely
$\{K_{1, 3}, 2P_1+P_2, 4K_1\}$-free graphs.

Observe that if $\{s_1, s_2, s_3\}$ and $\{t_1, t_2, t_3\}$ are two
vertex-disjoint triangles in a prismatic graph $G$, then there is a
perfect matching in $G$ between $\{s_1, s_2, s_3\}$ and
$\{t_1, t_2, t_3\}$, so that $\{s_1, s_2, s_3, t_1, t_2, t_3\}$
induces the complement of a $C_6$, commonly called
\emph{prism}, see Figure~\ref{f:prism}. This is where the name prismatic comes from.

\begin{figure}
  \begin{center}
\begin{tikzpicture}
\begin{scope}[xshift=0cm,yshift=0cm]
		\node[-] (a) at (0,0) {$s_1$};
		\node[-] (b) at (0,2) {$s_2$};
		\node[-] (c) at (1,1) {$s_3$};
		\node[-] (d) at (4,0) {$t_1$};
		\node[-] (e) at (4,2) {$t_2$};
		\node[-] (f) at (3,1) {$t_3$};
		\draw[->] (a) -- (b);
		\draw[<-] (a) -- (c);
		\draw[<-] (c) -- (b);
		\draw[->] (d) -- (e);
		\draw[<-] (d) -- (f);
		\draw[->] (e) -- (f);
		\draw[-] (a) -- (d);
		\draw[-] (b) -- (e);
		\draw[-] (c) -- (f);
\end{scope}
\end{tikzpicture}
\end{center}
\caption{The prism\label{f:prism}}
\end{figure}
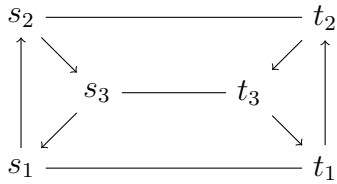

Chudnovsky and Seymour gave a full structural description of prismatic
graphs (and therefore of their complement).  They showed that the
class can be divided into two subclasses, to be defined later: the
orientable prismatic graphs, and the non-orientable prismatic graphs.
They described the structure of the two subclasses: orientable
in~\cite{chudnovsky_claw-free_2007} and non-orientable
in~\cite{chudnovsky_claw-free_2008}. Javadi and Hajebi \cite{DBLP:journals/jgt/JavadiH19} discovered a flaw in~\cite{chudnovsky_claw-free_2008}, that they could fix at the expense of adding a basic class in the structural description. 
 It is natural to ask whether
this description yields a polynomial time algorithm to color
antiprismatic graphs. 

The \emph{clique cover problem} is the problem of finding, in an input
graph $G$, a minimum number of cliques that partition $V(G)$.  It is
equivalent to the coloring problem for the complement. It is therefore
NP-complete in the general case. Our work is about the coloring
problem for antiprismatic graphs.  However, it is more convenient to
view it as a study of the clique cover problem for prismatic graphs.
Hence, from here on, we focus on the prismatic graphs and the clique
cover problem.

Our main result is that there exists an $O(n^{7.5})$-time algorithm to
solve the clique cover problem in non-orientable prismatic graphs.
Our proof is based on the existence, in any non-orientable prismatic
graph $G$ with more than $27$ vertices, of a set of at most $5$
vertices that intersects all triangles of $G$ (see
Theorem~\ref{l:27_ou_5}).  This follows directly from the structural
description, but needs careful verifications.  For the orientable
case, we could not settle the complexity of the clique cover problem,
but we prove that a related problem can be solved in polynomial
time: the \emph{vertex-disjoint triangles problem}.  It consists in finding a
maximum number of disjoint triangles in an input graph. This problem is known to be NP-hard in the general case \cite{hromkovic_vertex-disjoint_1998}.   

Our algorithm for the clique cover problem in the non-orientable case
relies on the existence of a hitting set of the triangles of bounded
size.  Sepehr Hajebi~\cite{hajebi} observed that the existence of such a set can be
proved with a short argument that relies on several lemmas
of~\cite{chudnovsky_claw-free_2008}.  This argument gives a hitting
set of size at most 15. The way we use hitting sets then provides an
algorithm for the clique cover problem of complexity
${\cal O}(n^{17.5})$.  

\subsection*{Outline}

In Section~\ref{sec:def} we give the main definitions that we need and
several results about prismatic graphs.

In Section~\ref{sec:no}, we give the structural description of
non-orientable prismatic graphs from~\cite{chudnovsky_claw-free_2008}
and show that it implies the existence of a set of bounded number of vertices 
that intersects all triangles. Since our proof mostly relies on the
structural description of Chudnovsky and Seymour, we have to give many
long definitions extracted from their work.

In Section~\ref{sec:col}, we show that
this yields a polynomial time algorithm for the clique cover problem in
non-orientable prismatic graphs.  

In Section~\ref{sec:o}, we prove that the vertex-disjoint triangles problem can
be solved in polynomial time for prismatic graphs since by Section~\ref{sec:col}. Our proof does not rely on the structural description
from~\cite{chudnovsky_claw-free_2007}.

In Section \ref{hajebi}, we describe Hajebi’s approach. 

Section \ref{s:concl} is devoted to concluding remarks.

\section{Definition, notation and prismatic graphs}
\label{sec:def}

Let $G$ be a graph. For $X \subseteq V(G)$, we denote by $G[X]$ the
subgraph of $G$ induced by $X$. For $v\in V(G)$, we denote by $N_G(v)$
the set of neighbours of $v$ in $G$ (note that $v\notin N_G(v)$); when there is no risk of ambiguity, we may write $N(v)$.   The
\textit{line graph} of $G$, denoted by $L(G)$, is the graph whose
vertices are the edges of $G$, and such that two vertices  are adjacent if
and only if they are adjacent edges of $G$.

A vertex $v$ is \emph{complete} to $A\subseteq V(G)$ if $v\notin A$
and $v$ is adjacent to all vertices of $A$.  When $A, B$ are two
disjoint subsets of $V(G)$, $A$ is \emph{complete} to $B$ if every
vertex in $A$ is complete to $B$. A vertex $v$ is \emph{anticomplete}
to $A\subseteq V(G)$ if $v\notin A$ and $v$ is adjacent to no vertex
in $A$.  When $A, B$ are two disjoint subsets of $V(G)$, $A$ is
\emph{anticomplete} to $B$ if every vertex in $A$ is anticomplete to
$B$.

A triangle in a graph $G$ is \textit{covered} by a set $S$ of vertices
if at least one vertex of the triangle is in $S$.   A set
$S\subseteq V(G)$ is a \textit{hitting set of the triangles of $G$}
if every triangle in $G$ is covered by $S$. We often write
\emph{hitting set} instead of hitting set of the triangles.

We call \emph{clique cover} of $G$ a set of disjoint cliques of $G$ that partition $V(G)$.

 \subsection*{Orientable and non-orientable prismatic graphs}

Let $T=\{a,b,c\}$ be a triangle in  a graph $G$. There are two cyclic
permutations of $T$, and we use the notation
$a \rightarrow b \rightarrow c \rightarrow a$ to denote the cyclic
permutation mapping $a$ to $b$, $b$ to $c$ and $c$ to $a$. Thus
$a \rightarrow b \rightarrow c \rightarrow a$ and
$b \rightarrow c \rightarrow a \rightarrow b$ mean the same
permutation.

A prismatic graph $G$ is \textit{orientable} if there is a choice of a
cyclic permutation $O(T)$ for every triangle $T$ of $G$, such that if
$S=\{s_1,s_2,s_3\}$ and $T=\{t_1,t_2,t_3\}$ are disjoint triangles and
$s_it_i$ is an edge for $1\leq i\leq 3$, then $O(S)$ is
$s_1\rightarrow s_2\rightarrow s_3 \rightarrow s_1$ if and only if
$O(T) $ is $t_1\rightarrow t_2 \rightarrow t_3 \rightarrow t_1$.  In that case, the set of permutations containing $O(T)$ for every triangle $T$ of  $G$ is called a \emph{correct orientation} of $G$.

 A graph $G$ is \textit{non-orientable} if there exists  no correct orientation of $G$.

Orientable and non-orientable prismatic graphs have very different
structures.  By Theorem~\ref{l:27_ou_5}, a non-orientable prismatic graph contains
at most 9 disjoint triangles.  It might seem surprising that having
a tenth triangle in a prismatic graph implies the existence of an
orientation.  This is because having a large number of disjoint
triangles in a prismatic graph entails so many constraints that the
only way to satisfy them all is in an orientable prismatic graph.

There is a nice characterisation of orientable prismatic
graphs.  The \emph{rotator} and the \emph{twister} are the graphs
represented on Figure~\ref{f:rt}. In the rotator there exists one triangle that intersects all triangles, we call it the \emph{ center of the rotator}.

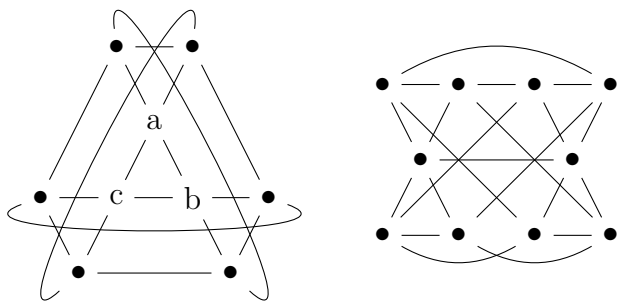
\begin{figure}[h]
 	\centering	
 	\begin{tikzpicture}
 \node[-] (a) at (0.5,-2) {$\bullet$};
 \node[-] (b) at (1,-3) {$\bullet$};
 \node[-] (c) at (1.5,-2) {c};
 \node[-] (d) at (1.5,0) {$\bullet$};
 \node[-] (e) at (2,-1) {a};
 \node[-] (f) at (2.5,0) {$\bullet$};
 \node[-] (g) at (2.5,-2) {b};
 \node[-] (h) at (3,-3) {$\bullet$};
 \node[-] (i) at (3.5,-2) {$\bullet$};

 \draw[-] (a) -- (b); 
 \draw[-] (a) -- (c);
\draw[-] (a) -- (d); 
 \draw[-] (a) to[bend right=160] (i);  
 \draw[-] (b) -- (c); 
 \draw[-] (b) to[bend left=160] (f); 
 \draw[-] (b) -- (h); 
 \draw[-] (c) -- (e); 
 \draw[-] (c) -- (g); 
\draw[-] (d) -- (e); 
 \draw[-] (d) -- (f); 
 \draw[-] (d) to[bend left=160](h); 
 \draw[-] (e) -- (f); 
 \draw[-] (e) -- (g); 
\draw[-] (f) --(i); 
 \draw[-] (g) -- (h);
  \draw[-] (i) -- (h);  
    \draw[-] (i) -- (g);

\node[-] (a2) at (5,-0.5) {$\bullet$};
 \node[-] (b2) at (5,-2.5) {$\bullet$};
 \node[-] (c2) at (5.5,-1.5) {$\bullet$};
 \node[-] (d2) at (6,-0.5) {$\bullet$};
 \node[-] (e2) at (6,-2.5) {$\bullet$};
 \node[-] (f2) at (7,-0.5) {$\bullet$};
 \node[-] (g2) at (7.5,-1.5) {$\bullet$};
 \node[-] (h2) at (7,-2.5) {$\bullet$};
 \node[-] (i2) at (8,-0.5) {$\bullet$};
  \node[-] (j2) at (8,-2.5) {$\bullet$};
  
   \draw[-] (c2) -- (g2); 

 \draw[-] (a2) -- (d2); 
 \draw[-] (a2) -- (c2);
  \draw[-] (d2) -- (c2);
\draw[-] (b2) -- (c2); 
 \draw[-] (b2)-- (e2);  
 \draw[-] (c2) -- (e2); 
 
 \draw[-] (f2) -- (i2); 
 \draw[-] (f2) -- (g2);
  \draw[-] (i2) -- (g2);
\draw[-] (h2) -- (g2); 
 \draw[-] (h2)-- (j2);  
 \draw[-] (g2) -- (j2);  
 
  \draw[-] (a2) to[bend left=30] (i2); 
 \draw[-] (a2)-- (h2);
  \draw[-] (d2) -- (j2);
\draw[-] (d2) -- (f2); 
 \draw[-] (b2)  to[bend right=30]  (h2);
 \draw[-] (b2) -- (f2);
 \draw[-] (e2)-- (i2);  
 \draw[-] (e2)  to[bend right=30]  (j2);
 \end{tikzpicture}
  
\caption{The rotator with center $\{a,b,c\}$ and the twister\label{f:rt}}  
\end{figure}

\begin{theorem}[$6.1$ in~\cite{chudnovsky_claw-free_2008}] \label{t:rotattor et twister}
  A prismatic graph is orientable if and only if it is $\{\mbox{rotator, twister}\}$-free.
\end{theorem}

The structure of non-orientable prismatic graphs presented in~\cite{chudnovsky_claw-free_2008} can be seen as a
description of how the rotator and the twister can be completed in order to obtain all non-orientable prismatic graphs.

\section{Non-orientable prismatic graphs}
\label{sec:no}

Our goal in this section is to prove the following.

 \begin{lemma}\label{supertheorem}
   Every prismatic non-orientable graph contains a hitting set of the
   triangles of cardinality smaller or equal to 10.
\end{lemma}

The \emph{core} of a graph $G$ is the union of all triangles of $G$.  Clearly,
in a prismatic graph, deleting an edge between two vertices that are
not in the core yields a prismatic graph. It follows that vertices not
in the core are less structured than vertices in the core.  Clearly,
to prove Lemma~\ref{supertheorem}, we may restrict our attention to
the cores of the graphs in the class. 

\medskip
 
A prismatic graph $G$ with core $W$ is \textit{rigid} if 
\begin{itemize}
\item there do not exist distinct $u,v\in V(G) \setminus W$ adjacent
  to precisely the same vertices in $W$, and
\item every two non-adjacent vertices of $G$ have a common neighbour in $W$.
\end{itemize}
 
\emph{Replicating} a vertex $v$ in a graph $G$ means
replacing $v$ by a stable set $S$ that is complete to $N(v)$ and
anticomplete to $V(G) \sm (N(v) \cup \{v\})$. We need the following. 

 \begin{theorem}[\label{rigidtoprism}2.2 from \cite{chudnovsky_claw-free_2008}]
   Every non-orientable prismatic graph can be obtained from a rigid
   non-orientable prismatic graph by replicating vertices not in the
   core, and then deleting edges between vertices not in the core.
 \end{theorem}

 It follows, from this result, that to prove Lemma~\ref{supertheorem}, it is enough to
 prove it for rigid non-orientable prismatic graphs.

 Theorem~$4.1$ in~\cite{chudnovsky_claw-free_2008} states  that
 the class of rigid non-orientable prismatic graphs is included in the
 union of 13 classes.  Javadi and Hajebi \cite{DBLP:journals/jgt/JavadiH19} discovered that one class is missing in Theorem 4.1,
the so-called class ${\cal F}_0$.  We describe these 14 classes whose union is called the
 \textit{menagerie}.

   In the
 definition of the menagerie, two operations are sometimes needed,
 the so-called \emph{multiplication} and \emph{exponentiation}.

The rest of the section is therefore organized as follows. The first
two subsections describe the multiplication and exponentiation
together with a proof that applying them under some specific
hypotheses preserves the existence of a hitting set. The next 14 subsections each presents one class of the menagerie, together with
a proof of the existence of a small  hitting set.  These subsections with Theorem~\ref{rigidtoprism} therefore form the proof of
Lemma~\ref{supertheorem}.

Before we start, we state the following  lemma which is a direct consequence of the definition of prismatic graphs.

\begin{lemma}\label{claimSchlafli}
If $v$ be a vertex of a prismatic graph $G$ then  $N_G(v)$ is a hitting set of $G$.
\end{lemma}

\subsection{Multiplication}

Let $H$ be a prismatic graph and $X$ be a subset of vertices of $H$.
For each vertex $x\in X$, let $A_x$ be a set of vertices not in $V(H)$
such that for all distinct $x, x'\in X$, $A_x \cap A_{x'}=\emptyset$. Let
$A=\cup _{x\in X} A_x$ and let $\varphi$ be a map from $A$ to the set of
integers such that for all $x\in X$, $\varphi$ is injective on $A_x$.

Let now $G$ be the graph defined as follows:

\begin{itemize}
\item $V(G)=(V(H)\setminus X)\cup A$.
\end{itemize}
Let $v$ and $v'$ be two distinct vertices of $G$.   
\begin{itemize}
\item If there is an $x\in X$ such that both $v$ and $v'$ are in $A_x$
  then $v$ and $v'$ are not adjacent.  
  $A_x$ is a stable set of $G$.
\item If $v$ and $v'$ are in $V(H)\setminus X$ then $vv'\in E(G)$ if
  and only if $vv'\in E(H)$. 
\item If $v\in V(H)\setminus X$ and $v'\in A_x$ for some $x\in X$
  then $vv'\in E(G)$ if and only if $vx\in E(H)$.  
\item If $v \in A_x$ and $v' \in A_{x'}$ where $x,x'\in X$ are distinct and adjacent in $H$, then $vv'\in E(G)$ if and only if
  $\varphi(v)=\varphi(v')$.
\item If $v \in A_x$ and $v' \in A_{x'}$ where $x,x'\in X$ are distinct and nonadjacent in $H$, then $vv'\notin E(G)$ if and only if
  $\varphi(v)=\varphi(v')$.
\end{itemize}

The graph $G$\textit{ is obtained from $H$ by multiplying} $X$. For
$x \in X$, the set $A_x$ is the \textit{set of new vertices
  corresponding to x} and $\varphi$ is the \textit{corresponding
  integer map}.  As noted in~\cite{chudnovsky_claw-free_2008}, the
multiplication does not preserve being prismatic in general, but it
is used only in situations where it does.

\begin{lemma}\label{multi}

 If $H[X]$ is an induced subgraph of $C_4$ and  non-adjacent
  vertices of $X$  have no common neighbours in $V(H)\setminus X$ then any hitting set of $H$  disjoint from $X$ is also a hitting set of $G$.
\end{lemma}

\begin{proof}
	Let $S_H$ be a hitting set of $S_H$.
  We prove that every triangle $\{u,v,w\}$ in $G$
  is covered by $S_H$. If $\{u,v,w\} \subseteq V(H)$, then it is
  covered by $S_H$, so we may assume that $|\{u,v,w\}\cap A|>0$.

  \medskip
  
  \noindent\textbf{Case I: $|\{u,v,w\}\cap A|=1$}

  Suppose up to symmetry that there exists $x\in X$ such that
  $u \in A_x$ and $v,w \in V(H)\setminus X$. Then $\{x,v,w\}$ is a
  triangle in $H$ and it has to be covered by $S_H$. Since
  $X\cap S_H= \emptyset$, $v$ or $w$ belongs to $S_H$.  Hence, $S_H$
  covers $\{u,v,w\}$.

  \medskip
  
  \noindent\textbf{Case II: $|\{u,v,w\}\cap A|=2$}

  Since for every $x\in X$, $A_x$ is a stable set, we may assume up to
  symmetry that there exist distinct $x, x' \in X$ such that
  $u \in A_x$, $v \in A_{x'}$ and $w \in V(H)\setminus X$.  Since $w$ is
  a common neighbour of $u$ and $v$ in $G$, $w$ is a common
  neighbour of $x$ and $x'$ in $H$.  From our assumptions, it follows that $x$
  and $x'$ are adjacent in $H$.  Hence $\{x,x',w\}$ is a triangle in $H$
  and it has to be covered by $S_H$. Since $X\cap S_H= \emptyset$, $w\in S_H$, we have $S_H$ covers $\{u,v,w\}$.

  \medskip
  
  \noindent\textbf{Case III: $|\{u,v,w\}\cap A|=3$} 

 Since for every $x \in X$, $A_x$ is a stable set, there exist then distinct $x, y, z \in X$ such that $u \in A_x$, $v \in A_y$  and $w \in
A_z$. Because of the hypothesis
  on $X$, $H[\{x,y,z\}]$ induces a $P_3$. Without loss of
  generality, suppose $xz\not\in E[H]$.

  Since $x$ and $y$ are adjacent in $H$, in order to have $u$ and $v$
  also adjacent in $G$, we have $\varphi (u)= \varphi (v)$.
  Similarly, $\varphi (v)= \varphi (w)$. Hence $\varphi(u)=\varphi (w)$.

  But since $x$ and $z$ are not adjacent, in order to have $u$ and
  $v$ adjacent in $G$, we need $\varphi (u) \neq \varphi (w)$, a
  contradiction.
\end{proof}

\subsection{Exponentiation}

A triangle $T=\{a,b,c\}$ of a graph $G$ is a \emph{ leaf triangle at c} if every triangle of $G$ distinct from $T$ contains neither a nor b.

Let $T=\{a,b,c\}$ be a leaf triangle at $c$ of a prismatic graph $H$. We define a partition of the neighbours of $c$, distinct from $a$ and $b$, into three disjoint sets: $D_1$, $D_2$, and $D_3$ as follows (see Figure~\ref{f:Nc}).  Let $v\neq a,b$ be a vertex adjacent to $c$ in $H$, then :

\begin{itemize}
\item $v\in D_1$ if $v$ belongs to a triangle that does not contain $c$.
\item $v\in D_2$ if $v\notin D_1$ and $v$ belongs to a triangle (then this triangle is unique and contains $c$). 
\item $v\in D_3$ if $v$ does not belong to any triangle.
\end{itemize}

\begin{figure}[h]
 	\centering	
 	\begin{tikzpicture}
 \node[-] (a) at (3.5,-1) {$a$};
 \node[-] (b) at (5.5,-1) {$b$};
 \node[-] (c) at (4.5,-3) {$c$};
 \node[-] (d) at (1,-7) {$\bullet$};
 \node[-] (e) at (2,-5) {$\bullet$};
 \node[-] (f) at (3,-7) {$\bullet$};
 \node[-] (g) at (4,-6) {$\bullet$};
 \node[-] (h) at (3,-8) {$\bullet$};
 \node[-] (i) at (5,-8) {$\bullet$};
 \node[-] (j) at (6,-5) {$\bullet$};
 \node[-] (k) at (8,-5) {$\bullet$};

 \draw[-] (a) -- (b); 
 \draw[-] (a) -- (c);
\draw[-] (b) -- (c); 
 \draw[-] (c) -- (e); 
 \draw[-] (c) -- (g); 
 \draw[-] (c) -- (j); 
 \draw[-] (c) -- (k); 
\draw[-] (d) -- (e); 
 \draw[-] (d) -- (f); 
 \draw[-] (e)-- (f); 
 \draw[-] (g) -- (h); 
 \draw[-] (i) -- (g); 
\draw[-] (i) -- (h); 
 \draw[-] (g) -- (j);
    
\node[-] (D1) at (1,-4) {$D_1$};    
\draw[dashed] (3,-5.5) ellipse (2cm and 1cm);

\node[-] (D2) at (6,-6) {$D_2$};    
\draw[dashed] (6,-5) ellipse (0.9cm and 0.5cm);

\node[-] (D3) at (8,-6) {$D_3$};    
\draw[dashed] (8,-5) ellipse (0.9cm and 0.5cm);
 \end{tikzpicture}
  
\caption{Neighbourhood of $c$ \label{f:Nc}}  
\end{figure}
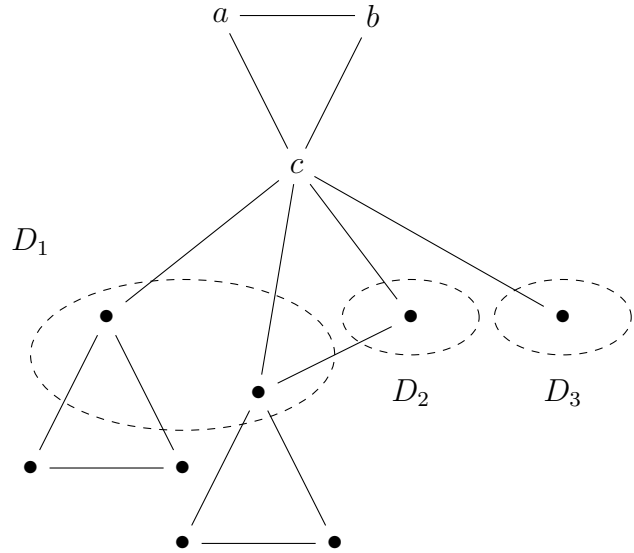

%

Let $A$, $B$ and $C$ be three pairwise disjoint sets of vertices. 

The graph $G$ is defined as follows:

\begin{itemize}
\item $V(G)=(V(H)\setminus \{a,b\})\cup A \cup B \cup C$
\end{itemize}
with the following adjacencies: 

\begin{itemize}
\item Vertices in $V(H)\sm \{a, b\}$ are adjacent in $G$ if and only
  if there are adjacent in $H$. 
\item $A$, $B$ and $C$ are stable sets.
\item Every vertex in $A$ has at most one neighbour in $B$ and vice
  versa.
\item Every vertex in $V(H)\setminus \{a,b\}$ adjacent (resp.\
  non-adjacent) to $a$ in $H$ is complete (resp.\ anticomplete) to $A$
  in $G$.
\item Every vertex in $V(H)\setminus \{a,b\}$ adjacent (resp.\
  non-adjacent) to $b$ in $H$ is complete (resp.\ anticomplete) to $B$
  in $G$.
\item $C$ is complete to $D_1 \cup D_3$ and anticomplete to
  $V(H)\setminus (\{a,b\}\cup D_1\cup D_3)$.
\item Every vertex in $C$ is adjacent to exactly one end of every edge
  between $A$ and $B$ and adjacent to every vertex in $A \cup B$ with
  no neighbour in $A\cup B$.

\end{itemize}

The graph $G$ is obtained from $H$ by \textit{exponentiating the leaf
  triangle $\{a,b,c\}$}.


Before proving that the exponentiation preserves hitting sets, note that  every prismatic graph $H$ with a leaf triangle $T=\{a,b,c\}$ at $c$ has a hitting set $S_H$
that contains $c$ but neither $a$ nor $b$.

\begin{lemma}\label{exp}
  If $G$ is prismatic and $S_H$ is a  hitting set of $H$ containing $c$ but neither $a$ nor $b$, then
  $S_H$ is also a hitting set of $G$.
\end{lemma}

\begin{proof}

  Let $\{u,v,w\}$ be a triangle in $G$. We show that one of $u$, $v$
  or $w$ belongs to $S_H$. Since $u,v,w\in V(G)$, none of them is $a$
  or $b$.
  

  \medskip
  
  \noindent\textbf{Case I:} $|\{u,v,w\}\cap (V(H)\setminus \{a,b\})|=3$

  This case is trivial because then $\{u,v,w\}$ is a triangle in $H$ and has to be covered by $S_H$.

  \medskip
  
  \noindent\textbf{Case II: $|\{u,v,w\}\cap (V(H)\setminus \{a,b\})|=2$}

  Without loss of generality suppose $u \notin V(H)$ and $v, w\in V(H)$.
 
  If $u$ belongs to $A$, that means that $\{a,v,w\}$ is a triangle in
  $H$. By our hypothesis this triangle should be T and $c\in \{u, v, w\}$. So
  $S_H$ covers $\{u, v, w\}$.  The case where $u$ belongs to $B$ is
  similar.

  If $u$ belongs to $C$, then $v$ and $w$ have to belong to
  $D_1 \cup D_3$ by definition of the neighbourhood of $C$. Then $\{c,v,w\}$ is a triangle in $G$.  So $\{u,v,w,c\}$ is a diamond in $G$, a contradiction to $G$ being prismatic.

  \medskip
  
  \noindent\textbf{Case III: $|(\{u,v,w\}\cap V(H)\setminus \{a,b\})|=1$}
  
	Without loss of generality suppose $w\in V(H)\setminus \{a\cup b\}$.
  Note that $A$, $B$ and $C$ are stable sets so $\{u,v,w\}$ contains at most
  one vertex of each. 

  Suppose that one of $u,v$ is in $C$ (so the other one is in $A\cup
  B$). Up to symmetry, we may assume
  that $u$ belongs to $A$ and $v$ belongs to $C$.  Then $aw\in E(H)$
  and $cw\in E(H)$. This means that
  $\{a,w,c\}$ is a triangle in $H$. This contradicts $\{a,b,c\}$ being a leaf
  triangle at $c$.
  
  Suppose that none of $u$, $v$ belong to $C$.  Up to symmetry,
  suppose $u\in A$ and $v\in B$.  So, $\{a, b, w\}$ is a triangle
  of $H$ and this triangle can only be $\{a,b,c\}$. So  $w=c\in S_H$.

  \medskip
  
  \noindent\textbf{Case IV: $|(\{u,v,w\}\cap H\setminus \{a,b\})|=0$ }

  This case cannot happen because $A$, $B$ and $C$ are stable sets
  and every vertex of $C$ has a unique neighbour in any edge of $G[A\cup B]$. 
\end{proof}

\subsection{Schläfli-prismatic graphs}
\label{Schladef}

We have to define the Schläfli graph, and it is more convenient to
work in the complement. The \emph{complement of the Schläfli graph}
 has 27 vertices $r^i_j,s^i_j,t^i_j$, $1\leq i,j \leq 3$ with
adjacencies as follows.  For $1 \leq i, i', j, j' \leq 3$:

\begin{itemize}
\item If $i\neq i'$ and $j\neq j'$, then $r^i_j$ is adjacent to
  $r^{i'}_{j'}$, $s^i_j$ is adjacent to $s^{i'}_{j'}$ and $t^i_j$ is
  adjacent to $t^{i'}_{j'}$.
\item If $j=i'$, then $r^i_j$ is adjacent to $s^{i'}_{j'}$, $s^i_j$ is
  adjacent to $t^{i'}_{j'}$ and $t^i_j$ is adjacent to $r^{i'}_{j'}$.
\end{itemize}
        
There are no other edges.
\\

This graph will be denoted by $\Sigma$ throughout the rest of the paper. We will often rely on the fact that $\Sigma$ is vertex-transitive. 

We introduce more notation. We set $R=\{r_j^i : 1\leq i,j \leq 3\}$,
$S=\{s_j^i : 1\leq i,j \leq 3\}$ and $T=\{t_j^i : 1\leq i,j \leq 3\}$
and call \emph{tile} each of the sets $R$, $S$, $T$. We call
\emph{line $i$} of $R$ the set $\{r^i_j : 1\leq j \leq 3\}$ and
\emph{column $j$} of $R$ the set $\{r^i_j : 1\leq i \leq 3\}$. We use
a similar notation for $S$ and $T$.

By definition an edge between $u$ and $v$ in a same tile exists if and only
if $u$ and $v$ are in different lines and columns. Edges between tiles
are conveniently described as follows: for every $i=1, 2, 3$, column $i$
of $R$ is complete to line $i$ of $S$, column $i$ of $S$ is complete
to line $i$ of $T$ and column $i$ of $T$ is complete to line $i$ of
$R$. There are no other edges.

A triangle in $\Sigma$ is \emph{internal} if it is included in a tile, and
\emph{external} otherwise.  The observations above show that an
internal triangle is made of three vertices that are in three
different lines, and also in three different columns of the tile. An
external triangle $\{u,v,w\}$ satisfies
$\{u,v,w\}=\{r^i_j,s^j_k,t^k_i\}$ for some $1\leq i,j,k\leq 3$. 

This shows that there exist $6$ internal triangles in each tile and
$27$ external triangles, that gives $45$ triangles in total. Each
vertex lies in two internal triangles and three external
triangles. Every edge is contained in exactly one triangle. 
See Figure~\ref{f:schafli}.

\begin{figure}[H]
 	\centering	
 	\begin{tikzpicture}
 \node[-] (r11) at (0,2) {$r_1^1$};
 \node[-] (r12) at (0,1) {$r_1^2$};
 \node[-] (r13) at (0,0) {$r_1^3$};
 \node[-] (r21) at (1,2) {$r_2^1$};
 \node[-] (r22) at (1,1) {$r_2^2$};
 \node[-] (r23) at (1,0) {$r_2^3$};
 \node[-] (r31) at (2,2) {$r_3^1$};
 \node[-] (r32) at (2,1) {$r_3^2$};
 \node[-] (r33) at (2,0) {$r_3^3$};
 
 \node[-] (s11) at (3,-1) {$s_1^1$};
 \node[-] (s12) at (3,-2) {$s_1^2$};
 \node[-] (s13) at (3,-3) {$s_1^3$};
 \node[-] (s21) at (4,-1) {$s_2^1$};
 \node[-] (s22) at (4,-2) {$s_2^2$};
 \node[-] (s23) at (4,-3) {$s_2^3$};
 \node[-] (s31) at (5,-1) {$s_3^1$};
 \node[-] (s32) at (5,-2) {$s_3^2$};
 \node[-] (s33) at (5,-3) {$s_3^3$};
 
 \node[-] (t11) at (-3,-2) {$t_1^1$};
 \node[-] (t12) at (-3,-3) {$t_1^2$};
 \node[-] (t13) at (-3,-4) {$t_1^3$};
 \node[-] (t21) at (-2,-2) {$t_2^1$};
 \node[-] (t22) at (-2,-3) {$t_2^2$};
 \node[-] (t23) at (-2,-4) {$t_2^3$};
 \node[-] (t31) at (-1,-2) {$t_3^1$};
 \node[-] (t32) at (-1,-3) {$t_3^2$};
 \node[-] (t33) at (-1,-4) {$t_3^3$};

 \draw[-] (r33) -- (r22); 
 \draw[-] (r33) -- (r21); 
 \draw[-] (r11) to[bend left=70] (r33); 
 
 \draw[-] (r11) -- (r22); 
 \draw[-] (r12) -- (r21); 
 \draw[-] (r12) -- (r33);

  \draw[-] (r33) to[bend left=80](s13); 
  \draw[-] (r33) to[bend left=80] (s23); `
  \draw[-] (r33) to[bend left=80] (s33); 
  
   \draw[-] (r33) -- (t31); 
   \draw[-] (r33) -- (t32); 
   \draw[-] (r33) -- (t33); 
   
   \draw[-] (s13) to[bend left=50] (t31); 
   \draw[-] (s23) to[bend left=50] (t32); 
   \draw[-] (s33) to[bend left=50] (t33); 
 \end{tikzpicture}
 \caption{ $\Sigma$, the complement of the Schläfli graph. (Only the 10 edges and the 5 triangles
   that contain $r_3^3$ are represented.)\label{f:schafli}}
 \end{figure}
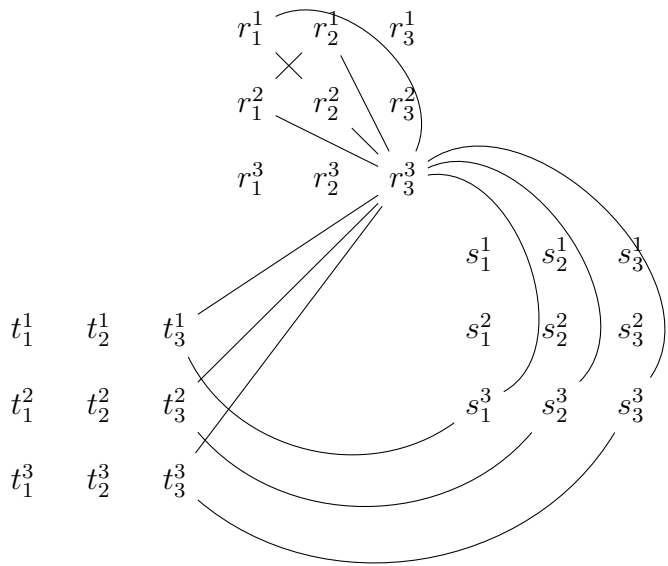	

 We call \textit{Schläfli-prismatic} graph every induced subgraph of
$\Sigma$. It is easy to see that they are prismatic.

%
%
%
%
%


\begin{lemma}\label{Sclacomp}
  A smallest hitting set of  $\Sigma$ has
  cardinality 10.
\end{lemma}

\begin{proof}
	Every vertex $v$ in $\Sigma$ has degree 10, so by Lemma \ref{claimSchlafli}, $N_\Sigma(v)$ is a hitting set of $\Sigma$ of size 10. 

  Suppose for a contradiction that $W$ is a hitting set of $\Sigma$ and $|W|=9$.  Since $\Sigma$ contains 45
  triangles and every vertex of   $\Sigma$ is contained in exactly 5 triangles (2
  internal and 3 external), no two vertices of $W$ hit the same
  triangle. Since every edge of  $\Sigma$ is contained in a triangle, it
  follows that $W$ is a stable set. 

  For each tile $X$, a maximum stable set in  $\Sigma$ has cardinality
  3 and is a line or a column. It follows that for $X\in \{R, S, T\}$,
  $W\cap X$ is a line or
  a column of $X$.

    By the pigeon hole principle, $W$ contains either two lines or two
  columns (of different tiles). This contradicts the fact that $W$
  is a stable set, because between two lines (or two columns) of
  different tiles, there exists at least one edge.

\end{proof}

 \subsection{Fuzzily Schläfli-prismatic graphs}

 Let $\{a,b,c\}$ be a leaf triangle at $c$ in a Schläfli-prismatic
 graph $H$. If a prismatic graph $G$ can be obtained from $H$ by multiplying $\{a,b\}$,
 and $A$, $B$ are the two sets of new vertices corresponding to $a$,
 $b$ respectively, the graph $G$ is \emph{ Schläfli-prismatic}. Note that this operation is not iterated.

\begin{lemma}\label{Fuzzily}
  Every fuzzily Schläfli-prismatic graph has a hitting set of
  cardinality smaller or equal to 5.
\end{lemma}

\begin{proof}
  Let $G$, $H$ and $\{a,b,c\}$ as in the definition.
  
  Since $a$ belongs to exactly one triangle of $H$ and to exactly 5 triangles in $\Sigma$, we have $|N_H(a)|\leq |N_\Sigma(a)|-4=6$.
  
  By Lemma \ref{claimSchlafli}, $N_H(a)$ is a hitting set of $H$.
  
  Since $b\in N_H(a)$ and since the unique triangle containing $b$ is $\{a,b,c\}$ which is already covered by $c$, we have that $N_H(a)\sm \{b\}$ is a hitting set of $H$ of cardinality at most 5. 
  
  By Lemma~\ref{multi} it is also a hitting set of $G$. 

%
%
%
%
%
\end{proof}

\subsection{Graphs of parallel-square type}\label{ss.parallel-square}

Let $X$ be the edge-set of some $C_4$ of the complete bipartite graph
$K_{3,3}$, and let $z$ be the edge of $K_{3,3}$ disjoint from all
edges in $X$. Thus $X$ induces a $C_4$ of the line graph $H$ of
$K_{3,3}$. Any graph $G$ obtained from $H$ by multiplying $X$, and
possibly deleting $z$, is prismatic and is called a graph of
\textit{parallel-square type}.

\begin{figure}[H]
 	\centering	
 	\begin{tikzpicture}
 \node[-] (v1) at (0,0) {$\bullet$};
 \node[color=red] (v2) at (0,2) {$\bullet$};
 \draw[color=violet] (0,2) circle (0.25cm);
 \node[color=red] (v3) at (1,1) {$\bullet$};
  \draw[color=violet] (1,1) circle (0.25cm);
 \node[color=blue] (v11) at (2,0.5) {$z$};
 \node[-] (v21) at (2,2.5) {$\bullet$};
 \node[-] (v31) at (3,1.5) {$\bullet$};
 \node[-] (v12) at (5,0) {$\bullet$};
 \node[color=red] (v22) at (5,2) {$\bullet$};
  \draw[color=violet] (5,2) circle (0.25cm);
 \node[color=red] (v32) at (4,1) {$\bullet$};
  \draw[color=violet] (4,1) circle (0.25cm);
 \node[color=red] (X) at (0.21,2.4) {$X$};
 \draw[-] (v1) -- (v2); 
 \draw[-] (v1) -- (v3); 
 \draw[color=red] (v3) -- (v2); 
 \draw[-] (v11) -- (v21); 
 \draw[-] (v11) -- (v31); 
 \draw[-] (v31) -- (v21); 
 \draw[-] (v12) -- (v22); 
 \draw[-] (v12) -- (v32); 
 \draw[color=red] (v32) -- (v22); 
 \draw[] (v1) -- (v11); 
 \draw[-] (v1) -- (v12); 
 \draw[-] (v11) -- (v12); 
 \draw[-] (v2) -- (v21); 
 \draw[color=red] (v2) -- (v22); 
 \draw[-] (v21) -- (v22); 
 \draw[-] (v3) -- (v31); 
 \draw[color=red] (v3) -- (v32); 
 \draw[-] (v31) -- (v32); 
 \end{tikzpicture}
 \caption{Line graph of $K_{3,3}$\label{f:spt}, $X$, $z$ as defined in section \ref{ss.parallel-square}} 
\end{figure}
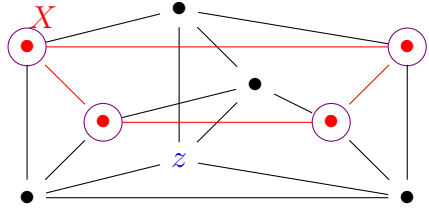
 
\begin{lemma}\label{l:parallel}
  Every prismatic graph of parallel square type admits a hitting set of
  cardinality smaller or equal to 4.
\end{lemma}

\begin{proof}
  Let $H$, $X$, $z$ and $G$ as in the definition.

  Let $S_H=V(H)\setminus (X\cup \{z\})$. Obviously $S_H$ is a hitting set of
  $H$.

  The set $X$ induces a $C_4$ in $H$  and no two
  non-adjacent vertices of $X$ have common neighbours in
  $V(H)\setminus X$. By Lemma~\ref{multi}, $S_H$ is a hitting set of
  $G$.

  Note that the deletion of $z$ does not change the result
  because $z$ is not in $S_H$.  Since $|S_H|=4$, the
  proof is completed.
\end{proof}

\subsection{Graphs of skew-square type}

Let $K$ be a graph with five vertices $a$, $b$, $c$, $s$, $t$, where
$\{s,a,c\}$ and $\{t,b,c\}$ are triangles and there is no more
edge. Let $H$ be obtained from $K$ by multiplying $\{a,b,c\}$, let $A$, $B$, $C$ be the sets of new vertices corresponding to $a$,
$b$, $c$ respectively, and let $\varphi$ be the corresponding integer
map.  Add three more vertices $d_1$, $d_2$, $d_3$ to $H$, with
adjacency as follows:

\begin{itemize}
\item $d_1$, $d_2$, $d_3$, $s$, $t$ are pairwise non-adjacent,
\item for $1\leq i \leq 3$ and $v\in A\cup B$, $d_i$ is adjacent to
  $v$ if and only if $1\leq \varphi(v)\leq 3$ and $\varphi(v) \neq i$,
\item for $1\leq i \leq 3$ and $v\in C$, $d_i$ is non-adjacent to $v$
  if and only if $1\leq \varphi(v)\leq 3$ and $\varphi(v) \neq i$.
\end{itemize} 

Any graph obtained by this way is prismatic, and is called a graph of  \textit{skew-square type}.

\begin{lemma}\label{l:skews}
  Every prismatic graph of skew-square type admits a hitting set of cardinality
  smaller or equal to 5.
\end{lemma}

\begin{proof}
  Let $G$ be a graph of skew-square type.

  Let $S_K=\{s,t\}$. We first show that
  $S_K$ is a hitting set of $H$ and then prove that
  $S_K\cup \{d_1,d_2,d_3\}$ is a hitting set of $G$.

  First it is obvious that $S_K$ is a hitting set of $K$. Furthermore, in $K$,
  $\{a,b,c\}$ induces a $P_3$ which is a induced subgraph of $C_4$  and vertices $a$ and $b$ do not have a common
  neighbour in $V(K)\setminus \{a,b,c\}$. We may therefore apply Lemma~\ref{multi}, showing that
  $S_K$ is a hitting set of $H$.

  Since we just add three vertices $d_1$, $d_2$ and $d_3$ to construct
  $G$ from $H$, every triangle in $G$ either contains one vertex of
  $\{d_1,d_2,d_3\}$ or is a triangle in $H$ that is covered by $S_K$.

  This shows that $S_K\cup \{d_1,d_2,d_3\}$ is a hitting set of
  $G$ of size~5.
\end{proof}

\subsection{The class ${\cal F}_{0}$}
Note that this class is defined in \cite{DBLP:journals/jgt/JavadiH19}.

 Let $H$ be a subgraph of  $\Sigma$ induced by:
$$\{r^i_j : (i,j) \in I_1\}\cup \{s^i_j : (i,j) \in I_2\}\cup \{t^i_j : (i,j) \in I_3\}$$

where $I=\{(i,j) : 1\leq i,j\leq 3\}$ and $I_1$, $I_2$, $I_3$ are subset of $I$ such that:
\begin{itemize}
\item $(1,1),(1,3),(2,2),(2,3),(3,1)(3,2)\in I_1$ and $(3,3)\notin I_1$,
\item $(1,1),(2,1),(3,2)\in  I_2 $ and $(1,2),(1,3),(2,2),(2,3)\notin I_2$,
\item $(1,3),(2,1),(2,2)\in  I_3$ and $(1,1),(1,2),(3,1),(3,2)\notin I_3$ .
\end{itemize}

Let $G$ be the graph obtained from $H$ by adding the edges $s^3_1 t^2_3$, $s^3_1 t_3^3$ and $s_3^3 t^2_3$ if the corresponding vertices are in $H$. We define ${\cal F}_{0}$ to be the class of all such graphs $G$. 

\begin{lemma}\label{l:f0}
  Every  graph of the class ${\cal F}_{0}$ admits a hitting set of
  cardinality smaller or equal to 3.
\end{lemma}

\begin{proof}
By Lemma~\ref{claimSchlafli}, $N_\Sigma (s^3_1)\cap V(H)=S_H=\{r^1_3,r^2_3, t^1_3\} $ is a hitting set of $H$. Note that $s^3_1$ and  $t_3^2$ do not have common neighbours in $H$ so as $s^3_1$ and $t^3_3$, and $s^3_3$ and $t^2_3$. Since $G[\{s^3_1,s^3_3,t_3^2,t^3_3\}]$ induces a $C_4$, we have that the addition of the new edges does not add any triangle in $G$. Therefore $S_H$ is also a hitting set of $G$. 
\end{proof}

\subsection{The class ${\cal F}_{1}$}

Let $G$ be a graph with vertex set the disjoint union of sets
$\{s,t\}$, $R$, $A$, $B$, where $|R|\leq 1$, and with edges as
follows:
\begin{itemize}
\item $s$, $t$ are adjacent, both are complete to $R$, and $s$ is complete to $A$; $t$ is complete to $B$;
\item every vertex in $A$ has at most one neighbour in $A$, and every
  vertex in $B$ has at most one neighbour in $B$;
\item if $a,a'\in A$ are adjacent and $b,b'\in B$ are adjacent, then
  the subgraph induced by $\{a,a',b,b'\}$ is a cycle;
\item if $a,a'\in A$ are adjacent and $b\in B$ has no neighbour in
  $B$, then $b$ is adjacent to exactly one of $a,a'$;
\item if $b,b'\in B$ are adjacent and $a\in A$ has no neighbour in
  $A$, then $a$ is adjacent to exactly one of $b,b'$;
\item if $a\in A$ has no neighbour in $A$, and $b \in B$ has no
  neighbour in $B$, then $a,b$ are adjacent
\end{itemize}

We define ${\cal F}_{1}$ to be the class of all such graphs $G$.

\begin{lemma}\label{l:f1}
Every prismatic graph of the class ${\cal F}_{1}$ admits a hitting set of cardinality smaller or equal to 2.
\end{lemma}

\begin{proof}
  We claim that $\{s,t\}$ is a hitting set
  of $G$.   This is equivalent to  the fact that in $G\sm \{s,t\}$, the
  neighborhood of any vertex $v$ is a stable set. And this follows
  directly from the definition in all cases ($v=r$, $v$ in $A$ with no
  neighbor in $A$, $v$ in $A$ with one neighbor in $A$, symmetric cases
  with $v\in B$).

\end{proof}

Note that graphs in ${\cal F}_1 $ can have arbitrarily large minimum degree.

\subsection{The class ${\cal F}_{2}$}
Let $K$ be the line graph of $K_{3,3}$ with vertices numbered $s_j^i$
$(1\leq i,j\leq 3)$, where $s_j^i$ and $s_{j'}^{i'}$ are adjacent if
and only if $i' \neq i$ and $j'\neq j$.  Note that this is how usually
 the complement of the line graph of $K_{3, 3}$ is defined, but
 since it is a self-complementary graph, it makes no difference.

 Let H be a graph obtained from this by multiplying
 $\{s^1_2,s^1_3,s_1^2,s_1^3\}$, thus, $H$ is of parallel-square
 type. Let $A_2^1$, $A_3^1$, $A^2_1$, $A_1^3$ be the sets of new
 vertices corresponding to $\{s^1_2$, $s^1_3$, $s_1^2$, $s_1^3\}$
 respectively, and let $\varphi$ be the corresponding integer
 map. Suppose that:

 \begin{itemize}
\item there do not exist $u\in A^3_1$ and $v\in A_3^1$ with
  $\varphi(u)=\varphi(v)$;
\item there exist $a_2^1 \in A_2^1$ and $a_1^2 \in A_1^2$ such that
  $\varphi (a_2^1) = \varphi(a_1^2) =1$;
\item $\varphi (v) \neq 1$ for all $v \in A_1^3 \cup A_3^1$.

\end{itemize}

Let $G$ be obtained from $H$ by exponentiating  $\{a^1_2,a_1^2,s_3^3\}$, leaf triangle at $s_3^3$.  We define ${\cal F}_{2}$ to be the class of all such graphs
$G$.

\begin{lemma}\label{l:f2}
  Every prismatic graph of the class ${\cal F}_{2}$ admits a hitting set of
  cardinality smaller or equal to 4.
\end{lemma}

\begin{proof}
   Let $S_K=\{s_2^3,s_3^2,s_2^2,s^3_3\}$. We can
  easily see that $S_K$ is a hitting set of $K$. We prove that $S_K$
  is a hitting set of $G$.

  By definition $H$ is obtained from $K$ by multiplying
  $\kappa =\{s_3^ 1,s_2^1,s_1^2,s_1^3\}$.  Note that, in $K$,  $\kappa$ induces a
  $C_4$, $\kappa \cap S_K = \emptyset$, $s_3^ 1,s_2^1$ do not
  have a common neighbours outside of $\kappa$ and $s_1^ 3,s_1^2$ do
  not have a common neighbours outside of $\kappa$. We may now apply
  Lemma~\ref{multi}, and conclude that $S_K$ is a hitting set of $H$.

We may apply Lemma \ref{exp}, and we obtain that $S_K$ is a hitting set of $G$ (note that the fact that  $H$ is a prismatic graph is not used in the proof of  Lemma~\ref{exp}) .
\end{proof}

 \subsection{The class ${\cal F}_{3}$}
 Let $K$ be the line graph of $K_{3,3}$, with vertices numbered $s_j^i$
 $(1\leq i,j \leq 3)$, where $s^i_j$ and $s^{i'}_{j'}$ are adjacent if
 and only if $i'\neq i$ and $j' \neq j$. Let $H$ be obtained from
 $K$ by deleting the vertex $s_2^2$ and possibly $s_1^1$, and then
 multiplying $\{s^1_2,s^1_3,s_1^2,s_1^3\}$. Let $A_2^1$, $A_3^1$,
 $A^2_1$, $A_1^3$ be the sets of new vertices corresponding to
 $s_2^1$, $s_3^1$, $s^2_1$, $s^3_1$ respectively, and let $\varphi$ be
 the corresponding integer map. Suppose that
 
 \begin{itemize}
 \item there exist $a_2^1 \in A_2^1$ and $a_1^3 \in A_1^3$ such that
   $\varphi (a_2^1) = \varphi(a_1^3) =1$;
\item $\varphi (v) \neq 1$ for all $v \in A_3^1 \cup A_1^2$.
\item there exist $a_3^1 \in A_3^1$ and $a_1^2 \in A_1^2$ such that
  $\varphi (a_3^1) = \varphi(a_1^2) =2$;
\item $\varphi (v) \neq 2$ for all $v \in A_2^1 \cup A_1^3$.
\end{itemize}

Let $G$ be obtained from $H$ by exponentiating 
$\{a^1_2,a_1^3,s_3^2\}$ and $\{a^1_3,a_1^2,s_2^3\}$, leaf triangles respectively at $s_3^2$ and $s_2^3$. We define ${\cal F}_{3}$ to be the
class of all such graphs $G$.
 
\begin{lemma}\label{l:f3}
  Every prismatic graph of the class ${\cal F}_{3}$ admits a hitting set of
  cardinality smaller or equal to 3.
\end{lemma}

\begin{proof}
  Let $S_H=\{s^3_3,s_2^3,s_3^2\}$. Note that in $K$ minus vertex $s_2^2$, $S_H$   is a hitting set,
  $X=\{s_3^1,s_1^3,s_2^1,s_1^2\}$ induces a $C_4$,
   $s_3^1,s_2^1$ do not share a common neighbour in $V(H)\sm X$ so as
  $s_1^2,s_1^3$. We may apply Lemma~\ref{multi}. It follows that
  $S_H$ is a hitting set of $H$.

  We may apply Lemma~\ref{exp}, and we
  conclude that $S_H$ is a hitting set of~$G$.
\end{proof}

 \subsection{The class ${\cal F}_{4}$}

 Take  $\Sigma$, with vertices numbered
 $r_j^i$, $s_j^i$, $t^i_j$ as usual. Let $H$ be the subgraph induced
 on $$Y \cup \{s^i_j : (i,j) \in I\} \cup \{t^1_1,t_2^2, t_3^3\}$$ where $\emptyset \neq Y \subseteq\{r_1^3,r_2^3,r_3^3\}$ and
 $I\subseteq \{(i,j) : 1\leq i,j\leq 3\}$ with $|I| \geq 8$ and
 including
 $\{(i,j) : 1\leq i \leq 3 \mbox{ and } 1\leq j \leq 2\}.$
 
 We consider $T=\{t_1^1,t_2^2,t_3^3\}$ which is a leaf triangle at $t_3^3$.  Let $G$ be obtained from $H$ by exponentiating 
$ T$. We define ${\cal F}_{4}$ to be the class of
 all such graphs $G$.

 \begin{lemma}\label{l:f4}
   Every prismatic graph of the class ${\cal F}_{4}$ admits a hitting set of
   cardinality smaller or equal to 4.
 \end{lemma}

\begin{proof}
 Let $K$ be the
  subgraph of  $\Sigma$ induced by the
  vertices:
  $\{r^3_1,r^3_2,r^3_3\}\cup S \cup \{t_1^1,t_2^2,t_3^3\}$.
 
  By Lemma~\ref{claimSchlafli},
  $N_\Sigma (r^3_3)\cap K=S_K=\{s^3_1,s^3_2,s_3^3,t_3^3\}$ is a hitting set of $K$.

  Since $H$ is a subgraph of $K$, $S_H=S_K\cap V(H)$ is a hitting set
  of $H$ and $|S_H| \leq 4$.

 Since $t_3^3\in S_H$, we can apply Lemma~\ref{exp}, and conclude that $S_H$
  is a hitting set of $G$.
\end{proof}

 \subsection{The class ${\cal F}_{5}$}
 Take  $\Sigma$, with vertices numbered
 $r^i_j$, $s^i_j$, $t^i_j$ as usual. Let $H$ be the subgraph induced
 on $$ \{r^i_j : (i,j) \in I_1\} \cup \{s^i_j : (i,j) \in I_2\} \cup \{t^i_j : (i,j) \in I_3\}$$
  where $I_1$, $I_2$, $I_3\subseteq \{(i,j) :1\leq i,j \leq 3\}$ are
 chosen such that:
 \begin{itemize}
 \item $(1,1)$, $(3,1)$, $(3,2)$, $(3,3)\in I_1$ and $(2,2)$, $(2,3)\notin I_1$
 \item $(1,1)\notin I_2$
 \item $(1,2)$, $(1,3)$, $(2,3)$, $(3,3)\in I_3$ and $(2,1)$, $(3,1)\notin I_3$
 \end{itemize}
 
Let $G$ be obtained from $H$ by adding the edge $r_1^1t_2^1$. We
 define ${\cal F}_{5}$ to be the class of all such graphs $G$.

 \begin{lemma}\label{l:f5}
   Every prismatic graph in the class ${\cal F}_{5}$  admits a hitting set of
   cardinality smaller or equals to 5.
 \end{lemma}
 
 \begin{proof}
 
 By Lemma~\ref{claimSchlafli},
 $N_\Sigma (r^1_1)\cap H=S_H = \{r^3_2,r^3_3,s^1_2,s^1_3,t^1_1 \}$ is a hitting set of $H$.
 
 In  $\Sigma$, vertices $r_1^1$ and $t_2^1$
 have as common neighbours  the following vertices: $s_1^1$, $r^2_2$, $r^2_3$, $t^2_1$ and
 $t_1^3$.  Since none of them are in $H$, the addition of the edge
 $r_1^1t_2^1$ does not create another triangle.
 
 This proves that $S_H$ is a hitting set of $G$.
 \end{proof}

 \subsection{The class ${\cal F}_{6}$}
 
Take  $\Sigma$ with vertices numbered $r^i_j$, $s^i_j$, $t_j^i$ as usual. Let
 $H$ be the subgraph induced by
 
 $$ \{r^i_j : (i,j) \in I_1\} \cup \{s^i_j : (i,j) \in I_2\} \cup \{t^i_j : (i,j) \in I_3\}$$
 
 where :
 \begin{itemize}
 \item $I_1=\{(1,1),(1,2),(3,1),(3,2),(3,3)\}$,
 \item $I_2=\{(1,2),(2,1),(2,2),(3,3)\}$,
 \item $I_3=\{(1,2),(2,2),(1,3),(2,3),(3,3)\}$
 \end{itemize}
 
 Let $G$ be obtained from $H$ by adding the edge $r_1^1 t_2^1$ and
 then multiplying $\{r^3_3,t_3^3\}$. We define ${\cal F}_{6}$ to be
 the class of all such graphs $G$.

 \begin{lemma}\label{l:f6}
   Every prismatic  graph in the class ${\cal F}_{6}$ admits a hitting set of
   cardinality smaller or equals to 3.
 \end{lemma}
 
 \begin{proof}
   By Lemma~\ref{claimSchlafli},
   $N_G (r^1_3)=\{r_1^3,r_2^3,s_3^3\}$ is a hitting set of $G$.
 \end{proof}

 \subsection{The class ${\cal F}_{7}$}
 
 The \textit{six-vertex prism} is the graph with six vertices
 $a_1, a_2, a_3,b_1,b_2,b_3$ and edges
$$a_1a_2, a_1a_3,a_2a_3,b_1b_2,b_1b_3,b_2b_3,a_1b_1,a_2b_2,a_3b_3$$

Let $K$ be a graph with six vertices, with the six-vertex prism as a
subgraph. Construct a new graph $G$ as follows. The vertices of $G$
consist of $E(K)$ and some of the vertices of $K$, so
$E(K) \subseteq V(G) \subseteq E(K) \cup V(K)$; two edges of $K$ are
adjacent in $G$ if they have no common end in $K$; an edge and a
vertex of $K$ are adjacent in $G$ if they are incident in $K$; and two
vertices of $K$ are adjacent in $G$ if they are non-adjacent in
$K$. The class of all such graphs $G$ is called ${\cal F}_{7}$ (they
are all prismatic). 
 
 \begin{lemma}\label{l:f7}
  Every prismatic graph in the class ${\cal F}_{7}$ admits a hitting set of
   cardinality at most $5$.
 \end{lemma}

 \begin{proof}
%
 Let $K$ and $G$ be as in the definition. 
Let us show that 

$$S_G=E(K)\cap \{a_1a_2,a_1a_3,a_1b_1,a_1b_2,a_1b_3\}$$ 

is a hitting set of $G$. Let $T$ be a triangle in $G$. We now break
into 4 cases.
 \begin{itemize}
 \item $|T \cap V(K)|=1$: Such a triangle does not exist. Because if two edges
   in $K$ are adjacent in $G$ then they do not share a common vertex in
   $K$ and then they do not have in $G$ a common neighbour in $V(G)\cap V(K)$.

 \item $|T\cap V(K)|=2$: Such a triangle does not exist. Indeed, if in $G$, there are two vertices $u,v\in V(K)$, both adjacent to $e\in E(K)$, then $e=uv$. A contradiction to  $u$, $v$ being adjacent in $G$.  
 	
 \item $|T\cap V(K)|=3$: Such a triangle does not exist. Otherwise it would induce in $K$ a stable set of cardinality 3 and there is no such stable set in the prism and hence in $K$.
 	
 \item $|T\cap V(K)|=0$: Every vertex of $G$ not in $V(K)$ is in
   $E(K)$. If such a triangle $T=\{e_1,e_2,e_3\}$, ($e_1,e_2,e_3\in E(K)$) exists, then $e_1$, $e_2$ and $e_3$ do not have common ends in $K$. It follows that at least one of $e_1$, $e_2$ or $e_3$ has $a_1$ as an end point and therefore is included in $S_G$.
 \end{itemize}

 Hence $S_G$ is a hitting set of $G$ and $|S_G|\leq 5$. 

\end{proof}

\subsection{The class ${\cal F}_{8}$}
 
Let $H$ be the graph with nine vertices $v_1, \dots , v_9$ and with
edges as follows: $\{v_1,v_2,v_3\}$ is a triangle, $\{v_4,v_5,v_6\}$
is complete to $\{v_7,v_8,v_9\}$, and for $i =1,2,3$, $v_i$ is
adjacent to $v_{i+3}$, $v_{i+6}$. Note that $H$ is a rotator. Let $G$ be obtained from $H$ by
multiplying $\{v_4,v_7\}$, $\{v_5,v_8\}$ and $\{v_6,v_9\}$. We define
${\cal F}_{8}$ to be the class of all such graphs $G$.

\begin{lemma}\label{l:f8}
  Every prismatic graph in the class ${\cal F}_{8}$ admits a hitting set of
  cardinality smaller or equals to 3.
\end{lemma}

 \begin{proof}
 We can easily see that $S_H=\{v_1,v_2,v_3\}$ is a hitting set of
 $H$.
  Since $\{v_4,v_7\}$, $\{v_5,v_8\}$ and $\{v_6,v_9\}$ are each
  not in $S_H$ and are each edges, we can successively apply Lemma~\ref{multi},
 for each one separately  and $S_H$ stays a
 hitting set of $G$ (note that the fact that  $H$ is a prismatic graph is not used in the proof of  Lemma~\ref{multi}).
 \end{proof}

 \subsection{The class ${\cal F}_{9}$}
 Take  $\Sigma$ with vertices numbered
 $r_j^i$, $s_j^i$, $t_j^i$ as usual. Let $H$ be the subgraph induced
 by
 $$\{r_j^i:(i,j)\in I_1\} \cup \{s_j^i:(i,j)\in I_2\}\cup \{t_j^i:(i,j)\in I_3\}$$
 
 where $I_1$, $I_2$, $I_3\subseteq \{(i,j) : 1\leq i,j \leq 3\}$ satisfy 
 
 \begin{itemize}
 \item $(2,1)$, $(3,1)$, $(3,2)$, $(3,3)\in I_1$ and $I_1$ contains at least one of $(1,2)$, $(1,3)$ and $(1,1)$, $(2,2)$, $(2,3)\notin I_1$,
 \item $(1,1)$, $(2,2)$, $(3,3)\in I_2$ and $(1,2)$, $(1,3)\notin I_2$,
 \item $(1,3)$, $(2,3)$, $(3,3)\in I_3$, and $I_3$ contains at least one of $(1,2)$, $(2,2)$, $(3,2)$, and $(1,1)$, $(2,1)$, $(3,1)\notin I_3$,
 \item either $(1,2)$, $(1,3)\in I_1$ or $I_3$ contains $(1,2)$ and at least one of $(2,2)$, $(3,2)$.
 \end{itemize}
 
 Let $G$ be obtained from $H$ by adding a new vertex $z$ adjacent to
 $r_2^3$, $r_3^3$, $s_1^1$, and to $t_2^2$ if $(2,2)\in I_3$, and to
 $t_2^3$ if $(3,2)\in I_3$. We define ${\cal F}_{9}$ to be the class
 of all such graphs $G$.

  \begin{lemma}\label{l:f9}
    Every prismatic graph in the class ${\cal F}_{9}$ admits a hitting set of
    cardinality smaller or equals to 3.
 \end{lemma}
 
 \begin{proof}
 By Lemma \ref{claimSchlafli},   $ N_G(r^1_1)=\{r^3_2,r^3_3,s^1_1\}$ is a hitting
  set of $G$.
 \end{proof}

 \section{Clique cover of non-orientable prismatic graphs}
\label{sec:col}

In this section we show how the presence of a hitting set of
cardinality bounded by a constant can be used for solving the clique
cover problem.  We have seen in the previous section that every
non-orientable prismatic graph admits a hitting set of size at most
10. The following is more useful for algorithmic purposes.

 \begin{theorem}\label{l:27_ou_5}
   If $G$ is a non-orientable prismatic graph then
   $G$ admits a hitting set of cardinality at  most $5$ or $G$ is a Schläfli-prismatic graph. 
 \end{theorem}

 \begin{proof}
   By Theorem \ref{rigidtoprism}, $G$ is
   obtained from a prismatic graph $H$ from the menagerie by
   replicating vertices not in the core of $H$ and then deleting edges
   between vertices not in the core.
  
  It is easy to verify that $G$ and $H$ have exactly the same
  triangles. Therefore, it is obvious that if $S_H$ is a hitting set
  of $H$ then $S_H$ is a hitting set of~$G$.

 From Lemmas~\ref{Fuzzily}, \ref{l:parallel},
  \ref{l:skews}, \ref{l:f0}, \ref{l:f1}, \ref{l:f2}, \ref{l:f3}, \ref{l:f4},
  \ref{l:f5}, \ref{l:f6}, \ref{l:f7}, \ref{l:f8}, and \ref{l:f9}, if
  $H$ is either Fuzzily Schläfli-prismatic, of parallel-square type,
  of skew-square type, or in ${\cal F}_i$,
  $i\in \{0,\dots , 9\}$, then $G$ admits a hitting set
  of size at most $5$.
   
 It remains to consider the case where $H$ is Schläfli-prismatic.  Hence, $H$ is an
  induced subgraph of  $\Sigma$.

  If every vertex of $H$ is contained in a triangle of $H$, then no vertex of $H$ can be replicated, so
  $G=H$ and $G$ is Schläfli-prismatic.  Hence, we may assume
  that some vertex $v$ of $H$ is contained in no triangle of $H$.

Since $N_{\Sigma}(v)$ induces a matching of 5
  edges in $\Sigma $ and $v$ is contained in no triangle of $H$, we see that
  $N_H(v)$ contains at most 5 vertices.
  By Lemma \ref{claimSchlafli}, $N_H(v)$ is a
  hitting set of $H$.  Hence, $H$ (and therefore $G$)
  contains a hitting set of size at most~5.
 \end{proof}
 
 Observe that in a triangle-free graph, solving the clique cover
 problem is easily reducible to computing a matching of maximum
 cardinality by Edmonds' algorithm.  Furthermore, the clique cover
 problem is solvable in constant time when the number of vertices of
 the input graph is bounded. Note that Schläfli-prismatic graphs have at most $27$ vertices.

 We need the following notation~: $T(G)$ is a variant of the adjacency
 matrix of $G$.  For $v, w \in V(G)$, the
 entry $(v, w)$  of $T(G)$ is $0$ if $v$ and $w$ are not adjacent,
 $1$ if they are adjacent but without common neighbour, and $x$ if
 they are adjacent and have $x$ as a common neighbour. Note that in
 the last case, if $G$ is  diamond-free and $K_4$-free, then $x$ is unique.

 This matrix can be computed in time ${\cal O}(n^3)$ at the beginning
 of an algorithm and used afterwards to find the triangles in $G$ or in any
 induced subgraph of $G$.

 \begin{lemma}\label{supersuper}
 	Let  $G \in $ Free$\{\mbox{diamond}, K_4\}$.
 	There is an algorithm that finds a hitting set of $G$ of cardinality at most $5$ if such a set exists and answers "no" otherwise. 
  This algorithm has complexity ${\cal O}(n^7)$.
 \end{lemma}
 
 \begin{proof}
 First compute $T(G)$ in time ${\cal O}(n^3)$ as above.
    Enumerate each
   set $X$ of vertices of $G$ of size at most 5 in time
   ${\cal O}(n^5)$.  For each $X$, test in time ${\cal O} (n^2)$ if
   $G\sm X$ is triangle-free by checking if every entry of $T(G)$
   reduced to $G\sm X$ is either $0$, $1$ or an element of $X$. If no such $X$ exists  answer "no" and  otherwise output $X$. 
 \end{proof}

 \begin{theorem}\label{t:algo_couverture}
 The Clique Cover Problem for
non-orientable prismatic graphs is solvable in  time ${\cal O}(n^{7.5})$.
 \end{theorem}
 
 \begin{proof}
   Let $G$ be a prismatic non-orientable graph.  The following method provides a minimum clique cover. 

   \begin{enumerate}
   \item Compute the matrix $T(G)$ as previously defined. This can be
     done in time ${\cal O}(n^3)$.
     
   \item Use the method in time ${\cal O} (n^{7})$ described in  Lemma \ref{supersuper}. If the algorithm outputs a hitting set of $G$ of size at most $5$  denoted by $S=\{s_1,\dots,s_{i^*}\}$ ($i^*\leq 5$), then go to Step~\ref{algo:oui}. Else by Theorem \ref{l:27_ou_5}, $G$ is a Schläfli-prismatic graph and go to Step~\ref{algo:no}.
     
     \item \label{algo:no} Since there is a bounded number of vertices in $G$ compute all possible clique covers of $G$ in constant time. Go to Step~\ref{algo:fin}.
     
   \item \label{algo:oui} Enumerate all sets $X$ of at most $5$ disjoint triangles of $G$, (this can trivially be done in time ${\cal O} (n^{15})$). We can do it in time ${\cal O} (n^{5})$ as follows.

Compute the set $\mathcal T_i$ of triangles containing $s_i$ for each $1\le i \le i^*$. This can be done in ${\cal O} (n)$ by reading the line of $T(G)$ corresponding to $s_i$. Notice that there are at most $n/2$ triangles in each $\mathcal T_i$. Then compute all subsets $\mathcal T$ of triangles of $G$ obtained by choosing at most one triangle in each $\mathcal T_i$.
 For each such $\mathcal T$ which contains only pairwise vertex-disjoint triangles, compute by some classical algorithm a maximum matching $M_{\mathcal T}$ of $G\sm (\cup_{T \in \mathcal T} T)$  and let $\mathcal R_{\mathcal T}$ be the vertices of $G$ that are neither in $\mathcal T$ nor in  $\mathcal M_{\mathcal T}$. Notice that $\mathcal T\cup \mathcal M_{\mathcal T}\cup \mathcal R_{\mathcal T}$ is a clique cover of $G$ . Go to Step~\ref{algo:fin}.

  \item\label{algo:fin}  Among all the clique covers generated by the previous steps, let $\mathcal C^*$ be one of the
   smallest size. Return $\mathcal C^*$.
\end{enumerate}

   \noindent{\bf Correctness:} 
   
   If $G$ has no hitting set of size at most $5$ then the algorithm will consider all possible clique covers of $G$. Therefore, the algorithm will give a clique cover of minimum size. 
   
  Otherwise,  let $\mathcal C$ be a minimum clique cover of $G$. Since $G$ is $K_4$-free,  $\mathcal C$ is the union of a set $\mathcal T $ of vertex-disjoint triangles, a set $\mathcal{E}$ of vertex-disjoint edges  and a set $\mathcal R$ of vertices. Since $\mathcal C$ is of minimum size, $\mathcal E$ should be  a maximum matching in the subgraph of  $G$ induced by the vertices not in any triangle  of $\mathcal T$.
   Each triangle in $\mathcal T$ contains a vertex of the hitting set $S$ of $G$ obtained by Step 2 in the algorithm. Furthermore, each vertex of $S$ is contained in at most one triangle of $\mathcal T$.  So the algorithm will consider $\mathcal T$ at some point in Step \ref{algo:oui}  and will compute a maximum matching in the remaining graph. Therefore it will return a clique cover $\mathcal C^*$ of same size as $\mathcal C$.

   \noindent{\bf Complexity}: 
   
   The procedure to enumerate all sets $\mathcal T$
   takes time ${\cal O}(n^{5})$.  For each
   such set, the maximum matching can be found by Micali and
   Vazirani's algorithm~\cite{conf/focs/MicaliV80} in time ${\cal O}(n^{2.5})$.  Overall, a best clique cover
   is found in time ${\cal O} (n^{7.5})$.
 \end{proof}

 \section{Orientable prismatic graphs}
 \label{sec:o}

	In the non-orientable case, we have shown that the clique cover problem is polynomial time solvable. We do not know the complexity of this problem in the orientable case. 
 	In this section we show that the vertex-disjoint triangles problem (the problem of finding a maximum number of vertex-disjoint triangles) is polynomial time solvable in prismatic graphs.  As noted in the introduction, this problem is NP-hard in the general case \cite{hromkovic_vertex-disjoint_1998}. 
	
Remark that solving the vertex-disjoint triangle problem is not sufficient to solve the clique cover problem in orientable prismatic graph. See Figure~\ref{f:mauvais}.

 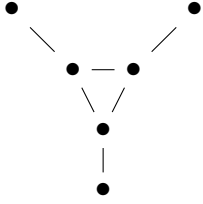
\begin{figure}[h]
\centering
	\begin{tikzpicture}[scale=0.8]
	\node[-] (a) at (1,4) {$\bullet$};
	\node[-] (b) at (2,3) {$\bullet$};
	\node[-] (c) at (2.5,2) {$\bullet$};
	\node[-] (d) at (2.5,1) {$\bullet$};
	\node[-] (e) at (3,3) {$\bullet$};
	\node[-] (f) at (4,4) {$\bullet$};
\draw[-] (a) -- (b);
\draw[-] (b) -- (c);
\draw[-] (b) -- (e);
\draw[-] (e) -- (c);
\draw[-] (e) -- (f);
\draw[-] (c) -- (d);
\end{tikzpicture}	
\caption{An orientable prismatic  graph where the best clique cover is not obtained by first selecting the maximum number of disjoint triangles }\label{f:mauvais}
\end{figure}

 The \textit{derived graph} $D(G)$ of a graph $G$ is the intersection graph of the triangles of $G$. More formally, if $H=D(G)$, then $V(H) = \{T : \mbox{ $T$ is a triangle in $G$}\}$. Two vertices of $H$ are adjacent if they are distinct triangles of $G$ sharing at least one vertex.  Note  that the class of derived graphs is not hereditary.

\begin{figure}[H]
\centering
	\begin{tikzpicture}[scale=0.8]
	\node[-] (a1) at (0,-2.5) {$a_3$};
\node[-](a2) at (1,-1) {$a_2$};
\node[-] (a3) at (2,-2.5) {$a_1$};
\node[-] (b1) at (4,-2.5) {$b_1$};
\node[-] (b2) at (5,-1) {$b_3$};
\node[-] (b3) at (6,-2.5) {$b_2$};
\node[-] (c1) at (3,-4) {$c_1$};
\node[-] (c2) at (2,-5.5) {$c_2$};
\node[-] (c3) at (4,-5.5) {$c_3$};
\node[-] (LK) at (3,-6.5) {$L(K_{3,3})$; (See also Fig.~\ref{f:spt})};

\node[-] (a) at (9,-2.5) {$a$};
\node[-](b) at (9,-3.5) {$b$};
\node[-] (c) at (9,-4.5) {$c$};
\node[-] (un) at (11,-2.5) {$1$};
\node[-] (d) at (11,-3.5) {$2$};
\node[-] (t) at (11,-4.5) {$3$};
\node[-] (name) at (10,-5.5) {$K_{3,3} $};

\draw[-] (a) -- (un);
\draw[-] (a) -- (d);
\draw[-] (a) -- (t);
\draw[-] (b) -- (un);
\draw[-] (b) -- (d);
\draw[-] (b) -- (t);
\draw[-] (c) -- (un);
\draw[-] (c) -- (d);
\draw[-] (c) -- (t);

\draw[-] (a1) -- (a2);
\draw[-] (a1) -- (a3);
\draw[-] (a2) -- (a3);
\draw[-] (b1) -- (b2);
\draw[-] (b1) -- (b3);
\draw[-] (b2) -- (b3);
\draw[-] (c1) -- (c2);
\draw[-] (c1) -- (c3);
\draw[-] (c2) -- (c3);
\draw[-] (a3) -- (b1);
\draw[-] (a3) -- (c1);
\draw[-] (b1) -- (c1);
\draw[-] (a1) -- (b2);
\draw[-] (b2) -- (c3);
\draw[-] (a1) -- (c3);
\draw[-] (a2) -- (b3);
\draw[-] (a2) -- (c2);
\draw[-] (b3) -- (c2);
\end{tikzpicture}	
\caption{$D(L(K_{3,3}))=K_{3,3}$}
		\end{figure}

\begin{theorem}\label{conect}
Let $G$ an orientable prismatic graph. Every connected component of $D(G)$ is claw-free or is isomorphic to $K_{3,3}$. 
\end{theorem}
 
\begin{proof}
Let $D$ be a connected component of $D(G)$  containing a  claw.  Hence, 
 $G$ has to contain 4 triangles as represented on Figure \ref{f:de_la_preuve} (not all edges of $G$ are represented).  We will use the notation given there and we denote by $K$ the set of vertices $\{a_1,a_2,a_3,b_1,b_2,b_3,c_1,c_2,c_3\}$.
 
\begin{figure}[h]
\centering
	\begin{tikzpicture}[scale=0.8]
	\node[-] (a1) at (1,-2.5) {$b_3$};
\node[-](a2) at (2,-1) {$b_2$};
\node[-] (a3) at (3,-2.5) {$b_1$};
\node[-] (a) at (2,-2) {B};
\node[-] (b1) at (5,-2.5) {$c_1$};
\node[-] (b2) at (6,-1) {$c_3$};
\node[-] (b3) at (7,-2.5) {$c_2$};
\node[-] (c1) at (4,-4) {$a_1$};
\node[-] (c2) at (3,-5.5) {$a_2$};
\node[-] (c3) at (5,-5.5) {$a_3$};
\node[-] (b) at (6,-2) {C};
\node[-] (c) at (4,-5) {A};

\draw[-] (a1) -- (a2);
\draw[-] (a1) -- (a3);
\draw[-] (a2) -- (a3);
\draw[-] (b1) -- (b2);
\draw[-] (b1) -- (b3);
\draw[-] (b2) -- (b3);
\draw[-] (c1) -- (c2);
\draw[-] (c1) -- (c3);
\draw[-] (c2) -- (c3);
\draw[-] (a3) -- (b1);
\draw[-] (a3) -- (c1);
\draw[-] (b1) -- (c1);
 
\end{tikzpicture}	
\caption{A graph whose derived graph is  $K_{1,3}$} \label{f:de_la_preuve}
		\end{figure}
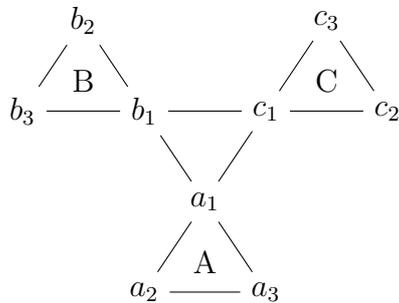

Since $G$ is prismatic there should be a matching between the extremities of the edge $a_2a_3$ and those of $b_2b_3$, $c_2c_3$.
 
Without loss of generality we may assume that these matching edges are: $a_2b_2$, $a_3b_3$ and $a_2c_2$, $a_3c_3$.

Now there are two possibilities for the matching between $b_2b_3$, $c_2c_3$. If $b_2c_3,b_3c_2\in E(G)$, then $G[K]$ is a rotator of center $\{a_1,b_1,c_1\}$, a contradiction to $G$ being orientable (Theorem \ref{t:rotattor et twister}). We can now assume that $b_2c_2,b_3c_3\in E(G)$. So $G[K]$ contains  $L(K_{3,3})$.

It remains to show that there is no other vertex in $D$. Assume that it is not the case, then there is a triangle $T$ in $G$, that intersects  $K$.  Since in $G[K]$, every vertex and edge is in a triangle,  $|T\cap K| = 1$.  

Since $L(K_{3, 3})$ is vertex transitive, we may assume up to symmetry that $T=\{u, v, b_2\}$.   Since $T$ and $\{a_1, b_1, c_1\}$ form a prism, we may assume without loss of generality that $ua_1$ and $vc_1$ are edges of $G$.  Now, because of triangles $T$ and $A$, $va_3 \in E(G)$ and because of triangles $T$ and $C$, $uc_3 \in E(G)$.  Then $\{u, v, b_2, a_2, a_3, a_1, c_2, c_3, c_1\}$ induces a rotator with center $\{a_2,b_2,c_2\}$, a contradiction to $G$ being  orientable. 
\end{proof}

\begin{theorem}
The vertex-disjoint triangle problem is ${\cal O}(n^5)$ time solvable in prismatic graphs. 
\end{theorem}

\begin{proof}
Let G be a prismatic graph.	If $G$ has at most 27 vertices, we solve the problem in constant time. Otherwise, we look by Lemma~\ref{supersuper}, for a hitting set of size 	at most 5 in $G$. If one exists, then we know that at most 5 disjoint triangles exist in $G$, and we find an optimal set of vertex-disjoint triangles of $G$ in time ${\cal O}(n^{5})$ as in the proof of Theorem~\ref{t:algo_couverture}. Hence, we may assume that no hitting set of size at most 5 exists. By Theorem~\ref{l:27_ou_5}, $G$ is orientable. 
	 
	 A set $R$ is a stable set of $D(G)$ if and only if it is a set of vertex-disjoint triangles of $G$.  Hence, it is enough to compute a maximum stable set in $D(G)$. Such a set can obviously be found by computing a maximum stable set in each connected component of $D(G)$.  By Lemma \ref{conect}, each such component is either isomorphic to $K_{3, 3}$ or claw-free.  The components that are isomorphic to $K_{3, 3}$ are handled trivially. In the components that are claw-free, to find a maximum stable set, we may rely on the classical algorithm of Sbihi~\cite{sbihi_algorithme_1980} (${\cal O}(n^3))$.  
\end{proof}

\section{A shorter proof for a weaker result}\label{hajebi}

Sepehr Hajebi~\cite{hajebi} noted that the existence of a hitting set
of bounded size (namely 15) in any non-orientable prismatic graph~$G$
can be deduced from several parts of~\cite{chudnovsky_claw-free_2008}
with a small amount of additional work as follows.

Consider a non-orientable prismatic graph $G$. By Theorem~$6.1$ in~\cite{chudnovsky_claw-free_2008}, $G$ contains either a twister or a rotator as an induced subgraph. 

If $G$ contains a twister $Z$ (so $|V(Z)|=10$) but does not contain any rotator, then (5) in the proof of 7.2
from~\cite{chudnovsky_claw-free_2008} shows that $V(Z)$ is a hitting set
of $G$.

If $G$ contains a rotator, then Hajebi's strategy is to rely on
results from Section 10 of~\cite{chudnovsky_claw-free_2008}.  This
section is about graphs that contain a rotator and no so-called
``square-forcer'' (we do not need the definition), but the argument
relies only on 10.3, where the assumption that there is no
square-forcer is not used.

From here on, we use notation from~\cite{chudnovsky_claw-free_2008}.

In 10.3, the set of all triangles of $G$ is described, and it is
proved that they all fall in one of the following categories: subsets
of $S$ (where $S$ is a set of cardinality at most 9), $R$-triangles, $T$-triangles,
diagonal triangles and marginal triangles.  The definition of these
categories implies that the set of vertices
${K}= \{r_1^3,r_2^3,r_3^3, t_3^1,t_3^2,t_3^3\} \cup S$ is 
a hitting set of $G$. It has size at most~15 (because $|S|\leq 9$).

\section{Concluding remarks}\label{s:concl}

Despite our efforts to use the results of \cite{chudnovsky_claw-free_2007},  the complexity of the clique cover problem remains unknown for orientable prismatic graphs. However, here is a simple remark. Suppose that $\{K_1, \dots, K_l\}$ is an optimal clique cover of a prismatic graph $G$. If for some $1\leq i, j\leq l$, $K_i$ is a triangle and $K_j$ an isolated vertex $v$, then by prismaticity, $v$ has a neighbour $u$ in $K_i$.  Hence, we may replace $K_i$ and $K_j$ by $K_i\sm \{u\}$ and $K_j\cup \{u\}$.  We may iterate this until in the optimal cover, there is not simultaneously a triangle and an isolated vertex.  It follows that for every prismatic graph, there exists an optimal cover that is either made only of triangles and edges, or made only of edges and isolated vertices. Observe that an optimal clique cover of  the last kind is easily computable in polynomial time by some classical algorithm for the maximum matching.  Hence, to solve the clique cover problem, it is enough to find an optimal clique cover of the first kind.

\section{Acknowledgement}

Thanks to \'Edouard Bonnet, Sepehr Hajebi and Stéphan Thomassé for useful discussions. 

\noindent This work was initiated by Fr\'ed\'eric Maffray.

 

\end{document}